\documentclass[10pt,conference]{IEEEtran}

\usepackage{amssymb}
\usepackage{amsmath}
\usepackage{epsfig}
\usepackage{color}

\newtheorem{theorem}{Theorem}
\newtheorem{corollary}{Corollary}

\DeclareMathOperator*{\modl}{mod}
\DeclareMathOperator*{\argmin}{argmin}

\begin{document}
\title{Hybrid Coding for Gaussian Broadcast Channels with Gaussian Sources}

\author{
\authorblockN{Rajiv Soundararajan}
\authorblockA{Department of Electrical \& Computer Engineering\\
University of Texas at Austin \\
Austin, TX 78712, USA\\
Email: rajiv.s@mail.utexas.edu}
\and
\authorblockN{Sriram Vishwanath}
\authorblockA{Department of Electrical \& Computer Engineering\\
University of Texas at Austin\\
Austin, TX 78712, USA\\
Email: sriram@ece.utexas.edu}
}

\maketitle

\begin{abstract}
This paper considers a degraded Gaussian broadcast channel over which Gaussian sources are to be communicated. When the sources are independent, this paper shows that hybrid coding achieves the optimal distortion region, the same as that of separate source and channel coding. It also shows that uncoded transmission is not optimal for this setting. For correlated sources, the paper shows that a hybrid coding strategy has a better distortion region than separate source-channel coding below a certain signal to noise ratio threshold. Thus, hybrid coding is a good choice for Gaussian broadcast channels with correlated Gaussian sources.\footnote{This work is supported by a grant from AT\&T Labs Austin, a grant from the Army Research Office and a grant from the Air Force Office of Sponsored Research.}
\end{abstract}

\section{Introduction}
\label{sec-introduction} 

The transmission of sources over a Gaussian broadcast channel \cite{Cover72} is a fundamental problem in information theory and arguably one of the better understood questions. In the case of independent sources over this degraded channel, the capacity region is characterized in \cite{Bergmans,Gallager1}. The achievable strategy for this channel in \cite{Bergmans} is superposition coding, but {\em dirty paper coding} \cite{Marton} can also be used to the same effect.

In contrast, the existing body of work on correlated sources over a broadcast channel is somewhat limited \cite{Han_Costa}. As source-channel separation does not hold, it is difficult to construct coding strategies and establish their optimality. Recently in \cite{Bross2008}, uncoded transmission of correlated Gaussians over a Gaussian broadcast channel was shown to be optimal below a signal to noise ratio (SNR) threshold . In related work, the transmission of a common source over a Gaussian broadcast channel with receiver side information was studied in \cite{Gunduz2008}.

In this work, we consider hybrid coding as a strategy for the Gaussian broadcast channel with or without correlated sources. By hybrid coding, we mean strategies that superimpose uncoded and coded transmission in communicating the sources to the destinations. Our hybrid coding strategy bears close resemblance to the dirty-paper-coding strategy using lattices, as developed in \cite{Erez_Shamai_Zamir}. We show that this hybrid strategy is optimal for the Gaussian broadcast channel with independent Gaussian sources. Extending it to the correlated case, we find that the strategy performs better (in terms of the distortion region achieved) than separate source and channel coding for SNRs below a certain threshold. 

In the next section, we present the system model. In Section \ref{sec:outerbound}, we present an outer bound on the  distortion region for this channel. We present the achievable scheme and the resulting distortion region for this channel in Section \ref{sec:achievable} and conclude with Section \ref{sec:conclude}.

\section{System Model}
\label{sec:system-model}
\begin{figure}[!th]
\centering
\scalebox{0.43}{
\input{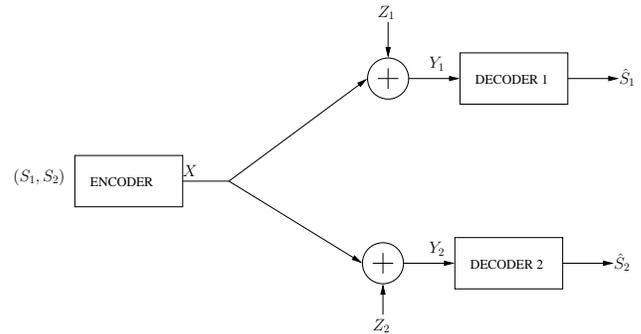}}
\caption{System Model}
\end{figure}
The system model is depicted in Fig. 1. Consider a sequence of independent and identically distributed (i.i.d) pair of correlated Gaussians $\{(S_{1}(i),S_{2}(i))\}_{i=1}^{n}$ with mean zero and covariance matrix,
\begin{equation*}
\mathbf{\Sigma}(i) = \left[\begin{array}{ccc} \sigma^{2} & \rho\sigma^{2} \\
\rho\sigma^{2} & \sigma^{2} \end{array}\right].
\end{equation*}
The goal is to transmit the pair over a degraded Gaussian broadcast channel to Receivers 1 and 2 respectively at the smallest possible distortion. We assume, with loss of generality, that $\rho>0$ and that the variances of $S_{1}(i)$ and $S_{2}(i)$ are equal. The transmitter applies an encoding function on the observed source sequence and transmits it over the channel. Mathematically,
\begin{equation*}
X^{n} = f(S^{n}_{1},S^{n}_{2})
\end{equation*}
where $S^{n}_{1}$ and $S^{n}_{2}$ denote $n$-length vectors. The transmitter is limited by an average second moment constraint on the channel input given by
\begin{equation*}
\frac{1}{n}\sum_{i=1}^{n}\mathbb{E}[(X(i))^{2}] \leq P.
\end{equation*}

The channel outputs at the two receivers are given by
\begin{align*}
Y_{1}(i) &= X(i)+Z_{1}(i)\\
Y_{2}(i) &= X(i)+Z_{2}(i)
\end{align*}
for $i=1,\ldots,n$, where $Z_{1}(i)$ and $Z_{2}(i)$ form an i.i.d sequence, independent of each other and are Gaussian distributed with mean zero and variance $N_{1}$ and $N_{2}$. Further, we assume that the broadcast channel is physically degraded with $N_{2}>N_{1}$. The receivers obtain estimates of the sources by applying a function on the received outputs. This is represented mathematically as 
\begin{equation*}
\hat{S}^{n}_{k} = \phi_{k}(Y^{n}_{k})
\end{equation*}
for $k=1,2$. The goal is to obtain estimates $\hat{S}^{n}_{1}$ and $\hat{S}^{n}_{2}$ within the minimum possible mean squared error. Therefore, we wish to obtain the smallest possible $D_{1}$ and $D_{2}$ where 
\begin{equation*}
D_{k} = \frac{1}{n}\sum_{i=1}^{n}\mathbb{E}[(S_{k}(i)-\hat{S}_{k}(i))^{2}]
\end{equation*}
for $k=1,2$. The distortion region $\mathcal{D}(\sigma^{2},\rho,P,N_{1},N_{2})$ is defined as the set of all pairs $(D_{1},D_{2})$ such that there exist encoding and decoding functions $f$, $\phi_{1}$ and $\phi_{2}$ resulting in distortions $D_{1}$ and $D_{2}$ at Receivers 1 and 2 respectively. Note that all logarithms are with respect to base 2 throughout the paper and $\mathbb{E}$ denotes the expected value of a random variable. 

\section{Outer Bound on Distortion Region}
\label{sec:outerbound}
We now present an outer bound on the conditional distortion region for the transmission of correlated Gaussian sources over a degraded broadcast channel. Let $\phi_{1|2}$ be the decoding function given the knowledge of both $Y^{n}_{1}$ and $S^{n}_{2}$ at Receiver 1. The conditional distortion region $\mathcal{D}_{c}(\sigma^{2},\rho,P,N_{1},N_{2})$, is defined as the set of all pairs $(D_{1|2},D_{2})$ such that there exist encoding function $f$ and decoding functions $\phi_{1|2}$ and $\phi_{2}$ resulting in distortions $D_{1|2}$ and $D_{2}$ at Receivers 1 and 2 respectively. The region described by Theorem 1 below is an alternative way of describing the outer bound region for the same channel presented in \cite{Bross2008}.\\

\begin{theorem} \label{theorem:outerbound} The conditional distortion region for transmission of correlated Gaussian sources over a degraded broadcast channel, $\mathcal{D}_{c}(\sigma^{2},\rho,P,N_{1},N_{2})$, consists of all pairs $(D_{1|2},D_{2})$ such that
\begin{equation*}
D_{1|2} \geq \frac{\sigma^{2}(1-\rho^{2})}{1+\frac{\alpha_{1}P}{N_{1}}} \textrm{ and }
D_{2} \geq \frac{\sigma^{2}}{1+\frac{(1-\alpha_{1})P}{\alpha_{1}P+N_{2}}}
\end{equation*}
where $\alpha_{1} \in [0,1]$. 
\end{theorem}

\begin{proof}
We first obtain a bound on the distortion $D_{1|2}$. By the data processing inequality (DPI), we have
\begin{equation}\label{eqn:dpi}
I(S^{n}_{1};\hat{S}^{n}_{1}|S^{n}_{2}) \leq I(S^{n}_{1};Y^{n}_{1}|S^{n}_{2}).
\end{equation}
The distortion in $S^{n}_{1}$ at Receiver 1 given that it knows $S^{n}_{2}$ and $Y^{n}_{1}$ is $D_{1|2}$ and the variance of $S^{n}_{1}$ given $S^{n}_{2}$ is $\sigma^{2}(1-\rho^{2})$. Since $S^{n}_{2}$ is known at both the transmitter and Receiver 1, 
\begin{equation}\label{eqn:sour1}
I(S^{n}_{1};\hat{S}^{n}_{1}|S^{n}_{2}) \geq \frac{n}{2}\log \frac{\sigma^{2}(1-\rho^{2})}{D_{1|2}}
\end{equation}
by definition of the rate distortion function for Gaussian sources \cite{Cover_Thomas}. Now, 
\begin{align}
I(S^{n}_{1};Y^{n}_{1}|S^{n}_{2}) &= h(Y^{n}_{1}|S^{n}_{2})-h(Y^{n}_{1}|S^{n}_{1},S^{n}_{2})\nonumber\\
&= \frac{n}{2}\log 2\pi e(\alpha_{1}P+N_{1})-h(Z^{n}_{1})\label{eqn:powalloc1}\\
&= \frac{n}{2}\log 2\pi e(\alpha_{1}P+N_{1})-\frac{n}{2}\log 2\pi eN_{1}\nonumber\\
&= \frac{n}{2}\log 2\pi e\bigg(1+\frac{\alpha_{1}P}{N_{1}}\bigg).\label{eqn:channel1}
\end{align}
The equality in (\ref{eqn:powalloc1}) results from the following argument. Since 
\begin{equation*}
\frac{n}{2}\log 2\pi eN_{1} \leq h(Y^{n}_{1}|S^{n}_{2}) \leq \frac{n}{2}\log 2\pi e(P+N_{1}),
\end{equation*}
there exists an $\alpha_{1}\in [0,1]$ such that 
\begin{equation}\label{eqn:powalloc}
h(Y^{n}_{1}|S^{n}_{2}) = \frac{n}{2}\log 2\pi e(\alpha_{1}P+N_{1}). 
\end{equation}
Therefore, from (\ref{eqn:sour1}) and (\ref{eqn:channel1}), we get
\begin{equation*}
D_{1|2} \geq \frac{\sigma^{2}(1-\rho^{2})}{1+\frac{\alpha_{1}P}{N_{1}}}.
\end{equation*}\\
\\
For source $S^{n}_{2}$, using DPI we observe that
\begin{equation}\label{eqn:dpis2}
I(S^{n}_{2};\hat{S}^{n}_{2}) \leq I(S^{n}_{2};Y^{n}_{2}).
\end{equation}
The rate distortion function for $S^{n}_{2}$ implies that 
\begin{equation}\label{eqn:source2}
\frac{n}{2}\log\frac{\sigma^{2}}{D_{2}} \leq I(S^{n}_{2};\hat{S}^{n}_{2}).
\end{equation}
Also, 
\begin{align}
I(S^{n}_{2};Y^{n}_{2}) &= h(Y^{n}_{2})-h(Y^{n}_{2}|S^{n}_{2})\nonumber\\
&\leq \frac{n}{2}\log 2\pi e(P+N_{2})-h(Y^{n}_{2}|S^{n}_{2})\label{eqn:gaussmax}\\
&\leq \frac{n}{2}\log 2\pi e(P+N_{2})-\frac{n}{2}\log 2\pi e(\alpha_{1}P+N_{2})\label{eqn:powalloc2}\\
&= \frac{n}{2}\log\bigg(1+\frac{(1-\alpha_{1})P}{\alpha_{1}P+N_{2}}\bigg),\label{eqn:channel2}
\end{align}
since in (\ref{eqn:gaussmax}), a Gaussian random variable maximizes entropy for a given variance and (\ref{eqn:powalloc2}) is true due to the following discussion. Note that due to the physically degraded nature of the broadcast channel, $Y^{n}_{2}$ may be written as $Y^{n}_{2}=Y^{n}_{1}+W^{n}$ where $W$ has variance $N_{2}-N_{1}$. Thus using (\ref{eqn:powalloc}) and entropy power inequality, we get
\begin{align*}
h(Y^{n}_{2}|S^{n}_{2}) &= h(Y^{n}_{1}+W^{n}|S^{n}_{2}) \\
&\geq \frac{n}{2}\log 2\pi e(\alpha_{1}P+N_{1}+N_{2}-N_{1})\\
&= \frac{n}{2}\log 2\pi e(\alpha_{1}P+N_{2}).
\end{align*}
Combining (\ref{eqn:dpis2}), (\ref{eqn:source2}) and (\ref{eqn:channel2}), we obtain
\begin{equation*}
D_{2} \geq \frac{\sigma^{2}}{1+\frac{(1-\alpha_{1})P}{\alpha_{1}P+N_{2}}}.
\end{equation*}
\end{proof}
Note that the outer bound on the conditional distortion region is obtained as a function of $\alpha_{1}$. We now state a corollary for the case of independent sources.

\begin{corollary} \label{corollary:indouterbound} The distortion region for transmission of independent Gaussian sources over a degraded broadcast channel, $\mathcal{D}(\sigma^{2},0,P,N_{1},N_{2})$, consists of all pairs $(D_{1},D_{2})$ such that
\begin{equation*}
D_{1} \geq \frac{\sigma^{2}}{1+\frac{\alpha_{1}P}{N_{1}}} \textrm{ and }
D_{2} \geq \frac{\sigma^{2}}{1+\frac{(1-\alpha_{1})P}{\alpha_{1}P+N_{2}}}
\end{equation*}
where $\alpha_{1} \in [0,1]$. 
\end{corollary}
\begin{proof}
We only present the proof for $D_{1}$ since the result for $D_{2}$ is the same as in the theorem above. Note that 
\begin{align}
I(S^{n}_{1};\hat{S}^{n}_{1}|S^{n}_{2}) &= h(S^{n}_{1}|S^{n}_{2})-h({S}^{n}_{1}|S^{n}_{2},\hat{S}^{n}_{1}) \nonumber\\
&=  h(S^{n}_{1})-h({S}^{n}_{1}|S^{n}_{2},\hat{S}^{n}_{1})\label{eqn:ind}\\
&\geq  h(S^{n}_{1})-h({S}^{n}_{1}|\hat{S}^{n}_{1})\label{eqn:condre}\\
& \geq \frac{n}{2}\log 2\pi e\frac{\sigma^{2}}{D_{1}}\label{eqn:source1}
\end{align}
where (\ref{eqn:ind}) uses the independence of $S^{n}_{1}$ and $S^{n}_{2}$, (\ref{eqn:condre}) is true because conditioning reduces entropy and (\ref{eqn:source1}) follows from the rate distortion function for Gaussian sources.\\
Now combining the above with (\ref{eqn:dpi}) and (\ref{eqn:channel1}), we get
\begin{equation*}
D_{1} \geq \frac{\sigma^{2}}{1+\frac{\alpha_{1}P}{N_{1}}}
\end{equation*}
\end{proof}

\section{Achievable Distortion Region}
\label{sec:achievable}
In this section, we present achievable distortion regions for transmitting independent and correlated Gaussian sources. We briefly discuss aspects of the coding scheme that are common for both the independent and correlated cases. The hybrid schemes proposed in the following subsections are based on lattices.  Let $\Lambda$ be a lattice of dimension $n$. Let the quantized value of $x\in \mathbb{R}^{n}$, $Q(x)=\argmin_{r\in \Lambda} \lVert x-r\rVert$. The fundamental Voronoi region of $\Lambda$ is defined as $\mathcal{V}_{0}=\{x\in\mathbb{R}^{n}: Q(x)=0\}$. Also, we denote $x \modl \Lambda = x - Q(x)$. We choose $\Lambda$ to be a `good' lattice for both source and channel coding \cite{Erez2005} and require it to have a second moment constraint $\sigma^{2}(\Lambda)=\frac{\int_{\mathcal{V}_{0}}\lVert x\rVert^{2}dx}{\int_{\mathcal{V}_{0}}dx} = P'$ where $P'$ will be specified later. Note that the transmitter has an average power constraint $P$. 

\subsection{Independent Gaussian Sources}
We now compare a hybrid coding strategy with uncoded transmission for communicating independent Gaussian sources. A hybrid coding scheme is basically a superposition of coded and uncoded transmission. Let the distortion region achieved by the hybrid scheme be  $\mathcal{D}_{h}(\sigma^{2},0,P,N_{1},N_{2})=\{(D_{1},D_{2}):D_{1} \textrm{ and } D_{2} \textrm{ are achieved by the hybrid scheme}\}$.\\ 
\begin{theorem}
$\mathcal{D}_{h}(\sigma^{2},0,P,N_{1},N_{2})$ = $\mathcal{D}(\sigma^{2},0,P,N_{1},N_{2})$
\end{theorem}
\begin{proof}
Consider a hybrid coding scheme in which the coded portion is given by
\begin{equation*}
X^{n}_{1} = [S^{n}_{1}+\beta\gamma S^{n}_{2}+U^{n}]\modl\Lambda,
\end{equation*}
where $\beta$ and $\gamma$ are constants which will be specified later and $U^{n}$ is the dither which is known apriori to both the transmitter and receivers and is uniformly distributed in $\mathcal{V}_{0}$. We send $S^{n}_{2}$ uncoded after scaling it appropriately to meet the power constraint. In the following, $\alpha_{1}$ represents the power allocation in the outer bound discussion. The channel input is a superposition of coded and uncoded transmission and is expressed as  
\begin{equation*}
X^{n} = \alpha X^{n}_{1}+\gamma S^{n}_{2},
\end{equation*}
where $\alpha$ satisfies \begin{equation}\label{eqn:poweq}
\alpha^{2} P' = \alpha_{1} P
\end{equation}
and
\begin{equation}\label{eqn:gamma} 
\gamma = \sqrt{\frac{(1-\alpha_{1})P}{\sigma^{2}}}.
\end{equation}
Note that $X^{n}_{1}$ and $S^{n}_{2}$ are independent of each other on account of addition of the uniform dither before the modulo operation. We also observe that the scheme is similar to the dirty paper coding strategy in \cite{Erez_Shamai_Zamir} where $S^{n}_{2}$ resembles the interference known at the transmitter.
The output at the receivers is given by 
\begin{equation}\label{eqn:output}
Y^{n}_{k} = X^{n}+Z^{n}_{k}
\end{equation}
for $k=1,2$. \\
\\
At Receiver 1, we perform the following series of operations: 
\begin{align}
R^{n}_{1} =& [\delta Y^{n}_{1}-U^{n}]\modl\Lambda \label{eqn:start}\\
=& [(\delta\alpha-1)X^{n}_{1}+S^{n}_{1}+(\beta\gamma+\delta\gamma)S^{n}_{2}+\delta Z^{n}_{1}]\modl\Lambda\\
=& [S^{n}_{1}+\gamma(\delta+\beta)S^{n}_{2}+(\delta\alpha-1)X^{n}_{1}+\delta Z^{n}_{1}]\modl\Lambda\\
=& [S^{n}_{1}+W^{n}_{1}]\modl\Lambda \label{eqn:end}
\end{align}
where $W^{n}_{1} = \gamma(\delta+\beta)S^{n}_{2}+(\delta\alpha-1)X^{n}_{1}+\delta Z^{n}_{1}$ is the effective noise term independent of $S^{n}_{1}$. Choosing $\beta=-\delta$ and $\delta = \frac{\alpha P'}{\alpha^{2} P'+N_{1}}$, we reduce the variance of the effective noise to $\frac{P'N_{1}}{\alpha^{2}P'+N_{1}}$. Since $\Lambda$ is a `good' channel lattice, we require $P'$ to satisfy
\begin{equation}\label{eqn:corrdec}
\sigma^{2} + \frac{P'N_{1}}{\alpha^{2}P'+N_{1}} \leq P',
\end{equation}
for correct decoding with high probability as $n\rightarrow\infty$ \cite{Erez2004}\cite{Kochman2008}. This leads to 
\begin{equation*}
R^{n}_{1} = [S^{n}_{1}+W^{n}_{1}]\modl\Lambda = S^{n}_{1}+W^{n}_{1}.
\end{equation*}
Allowing $P'$ to satisfy (\ref{eqn:corrdec}) with equality and using (\ref{eqn:poweq}), we obtain
\begin{equation*}
P' = \sigma^{2}\frac{\alpha_{1}P+N_{1}}{\alpha_{1}P} \textrm{ and } \alpha=\sqrt{\frac{\alpha_{1}P}{P'}}.
\end{equation*}
Thus $W_{1}$ has variance $\frac{P'N_{1}}{\alpha^{2}P'+N_{1}} = \frac{\sigma^{2}N_{1}}{\alpha_{1}P}$. The estimator 
\begin{equation*}
\hat{S}^{n}_{1} = \frac{\sigma^{2}}{\sigma^{2}+\sigma^{2}\frac{N_{1}}{\alpha_{1}P}}R^{n}_{1},
\end{equation*}
achieves a distortion in $S_{1}$ of 
\begin{equation*}
D_{1} = \sigma^{2} - \frac{\sigma^{4}}{\sigma^{2}+\sigma^{2}\frac{N_{1}}{\alpha_{1}P}}= \frac{\sigma^{2}}{1+\frac{\alpha_{1}P}{N_{1}}}.
\end{equation*}
On the other hand, Receiver 2 observes
\begin{equation*}
Y^{n}_{2} = \alpha X^{n}_{1}+\gamma S^{n}_{2}+Z^{n}_{2} = \gamma S^{n}_{2}+W^{n}_{2}, 
\end{equation*}
where $W^{n}_{2}=\alpha X^{n}_{1}+Z^{n}_{2}$ is treated as the effective noise which is independent of $S^{n}_{2}$. We now construct the linear estimator 
\begin{equation*}
\hat{S}^{n}_{2} = \frac{\gamma\sigma^{2}}{P+N_{2}}Y^{n}_{2}.
\end{equation*}
Using (\ref{eqn:gamma}), this estimator results in a distortion 
\begin{equation*}
D_{2} = \frac{\sigma^{2}}{1+\frac{(1-\alpha_{1})P}{\alpha_{1}P+N_2}}.
\end{equation*}
Thus the superposition scheme described above achieves the optimal distortions for sources $S_{1}$ and $S_{2}$. 
\end{proof}

We now show that uncoded transmission is sub-optimal for independent sources through the following theorem. Let the distortion region achieved by uncoded transmission be $\mathcal{D}_{u}(\sigma^{2},0,P,N_{1},N_{2})$. \\
\begin{theorem}
$\mathcal{D}_{u}$ is equal to the set of all distortion pairs $(D_{1},D_{2})$ such that 
\begin{equation*}
D_{1} \geq \frac{\sigma^{2}}{1+\frac{\alpha_{1}P}{(1-\alpha_{1})P+N_{1}}} \textrm{ and }
D_{2} \geq \frac{\sigma^{2}}{1+\frac{(1-\alpha_{1})P}{\alpha_{1}P+N_{2}}}
\end{equation*}
where $\alpha_{1}\in[0,1]$. Further,
\begin{equation*}
\mathcal{D}_{u}(\sigma^{2},0,P,N_{1},N_{2}) \subset \mathcal{D}(\sigma^{2},0,P,N_{1},N_{2}). 
\end{equation*}
\end{theorem}
\begin{proof}
The uncoded transmission strategy is to send a linear combination of $S^{n}_{1}$ and $S^{n}_{2}$. The power allocated for sending $S^{n}_{1}$ and $S^{n}_{2}$ are $\alpha_{1}P$ and $(1-\alpha_{1})P$ respectively. Therefore we transmit,
\begin{equation*}
X^{n} = \sqrt{\frac{\alpha_{1}P}{\sigma^{2}}}S^{n}_{1}+\sqrt{\frac{(1-\alpha_{1})P}{\sigma^{2}}}S^{n}_{2}.
\end{equation*}
Receiver 1 obtains a minimum mean squared error (MMSE) estimate of $S^{n}_{1}$ given $Y^{n}_{1}$ as 
$\hat{S}^{n}_{1} = \frac{\sqrt{\alpha_{1}P\sigma^{2}}}{P+N_{1}}Y^{n}_{1}$. This leads to a distortion in $S_{1}$,
\begin{equation*}
D_{1} = \frac{\sigma^{2}}{1+\frac{\alpha_{1}P}{(1-\alpha_{1})P+N_{1}}}.
\end{equation*}
The estimate of $S^{n}_{2}$ is given by $\hat{S}^{n}_{2} = \frac{\sqrt{(1-\alpha_{1})P\sigma^{2}}}{P+N_{2}}Y^{n}_{2}$ resulting in 
\begin{equation*}
D_{2} = \frac{\sigma^{2}}{1+\frac{(1-\alpha_{1})P}{\alpha_{1}P+N_{2}}}.
\end{equation*}
We observe that while uncoded transmission scheme achieves the optimal distortion in $S_{2}$, the distortion in $S_{1}$ is higher than the optimal distortion. Thus $\mathcal{D}_{u}(\sigma^{2},0,P,N_{1},N_{2}) \subset \mathcal{D}(\sigma^{2},0,P,N_{1},N_{2})$.
\end{proof}

\subsection{Correlated Gaussian Sources}
Next, we extend our hybrid coding scheme to the problem of correlated sources, with the goal of achieving a better distortion region than source channel separation. We consider two source channel separation schemes in this section, Scheme A and Scheme B. Scheme A treats the messages obtained by compressing $S_{1}$ and $S_{2}$ as independent and communicates them reliably over the broadcast channel. Scheme B explores the idea of using Wyner Ziv coding for communicating $S_{1}$. Let the distortion region achieved by the hybrid scheme be $\mathcal{D}_{h}(\sigma^{2},\rho,P,N_{1},N_{2})$ and the distortion region achieved by Scheme A and Scheme B be $\mathcal{D}_{A}(\sigma^{2},\rho,P,N_{1},N_{2})$ and $\mathcal{D}_{B}(\sigma^{2},\rho,P,N_{1},N_{2})$.\\
\begin{theorem}
If $\alpha_{1}<\frac{1}{2}$ and $\frac{P}{N_{1}} < \frac{1-2\alpha_{1}}{\alpha^{2}_{1}}$, then
\begin{equation*}
\mathcal{D}_{A}(\sigma^{2},\rho,P,N_{1},N_{2}) \subset \mathcal{D}_{h}(\sigma^{2},\rho,P,N_{1},N_{2}). 
\end{equation*}
For any $\frac{P}{N_{1}}>0$ and $0\leq\alpha_{1}\leq1$,
\begin{equation*}
\mathcal{D}_{B}(\sigma^{2},\rho,P,N_{1},N_{2}) \subset \mathcal{D}_{h}(\sigma^{2},\rho,P,N_{1},N_{2}).
\end{equation*}
\end{theorem}
\begin{proof}
We begin by noting that $S^{n}_{1}$ may be represented as $S^{n}_{1}=\rho S^{n}_{2}+V^{n}$ with $V^{n}$ independent of $S^{n}_{2}$ and Gaussian distributed with mean zero and variance $\sigma^{2}(1-\rho^{2})$. The main idea of the hybrid coding scheme is to use the scheme proposed in the previous subsection to send $V^{n}$ and $S^{n}_{2}$, which are independent. Thus, the transmitter forms the coded portion of the channel input similar to the independent case using the lattice $\Lambda$ as 
\begin{equation*}
X^{n}_{1} = [V^{n}+\beta\gamma S^{n}_{2}+U^{n}]\modl\Lambda. 
\end{equation*}
As before, $S_{2}^{n}$ is sent uncoded and superposed on the coded portion after an appropriate scaling to satisfy the power constraint. Thus the channel input is given by 
\begin{equation*}
X^{n} = \alpha X^{n}_{1}+\gamma S^{n}_{2},
\end{equation*}
where $\alpha$ and $\gamma$ satisfy (\ref{eqn:poweq}) and (\ref{eqn:gamma}) respectively. \\
\\
The channel output at Receiver 1 is expressed as 
\begin{equation}\label{eqn:y1}
Y^{n}_{1} = \alpha X^{n}_{1}+\gamma S^{n}_{2}+Z^{n}_{1}.
\end{equation}
The receiver can now perform the same sequence of operations as in Equations (\ref{eqn:start})-(\ref{eqn:end}) to obtain 
\begin{equation*}
R^{n}_{11} = [V^{n}+W^{n}_{11}]\modl\Lambda 
\end{equation*}
where $W^{n}_{11}=\gamma(\delta+\beta)S^{n}_{2}+(\delta\alpha-1)X^{n}_{1}+\delta Z^{n}_{1}$ represents the effective noise independent of $V^{n}$. By choosing $P' = \sigma^{2}(1-\rho^{2})\frac{\alpha_{1}P+N_{1}}{\alpha_{1}P}$, $\alpha=\sqrt{\frac{\alpha_{1}P}{P'}}$, $\delta = \frac{\alpha P'}{\alpha_{1}P+N_{1}}$ and $\beta=-\delta$, we get
\begin{equation}\label{eqn:vnoisy}
R^{n}_{11}=V^{n}+W^{n}_{11},
\end{equation}
where the variance of $W_{11}$ is given by $\frac{\sigma^{2}(1-\rho)N_{1}}{\alpha_{1}P}$. Observe that we can rewrite (\ref{eqn:y1}) as a noisy observation of $S^{n}_{2}$ in the form
\begin{equation}\label{eqn:snoisy}
R^{n}_{12}=S^{n}_{2}+W^{n}_{12}
\end{equation}
where $W^{n}_{12}=\frac{(\alpha X^{n}_{1}+Z^{n}_{1})}{\gamma}$ is independent of $S^{n}_{2}$ and has variance $\frac{(\alpha_{1}P+N_{1})\sigma^{2}}{(1-\alpha_{1})P}$. \\
\\
Now, using the noisy observations of $V^{n}$ and $S^{n}_{2}$ in (\ref{eqn:vnoisy}) and (\ref{eqn:snoisy}) respectively, we construct a linear estimator of $S^{n}_{1}$ given $R^{n}_{11}$ and $R^{n}_{12}$. Before finding the estimator, we observe that $W^{n}_{11}$ and $W^{n}_{12}$ are uncorrelated for the choice of constants $\alpha$ and $\delta$ stated above. For each time instant $i$,
\begin{align*}
\mathbb{E}[W_{11i}W_{12i}]&= \mathbb{E}[((\delta\alpha-1)X_{1i}+\delta Z_{1i})\frac{(\alpha X_{1i}+Z_{1i})}{\gamma}]\\
&=\frac{1}{\gamma}\Big((\delta\alpha-1)\alpha P'+\delta N_{1}\Big)=0
\end{align*}
since $\delta=\frac{\alpha P'}{\alpha^{2}P'+N_{1}}$. Therefore $R^{n}_{11}$ and $R^{n}_{12}$ are uncorrelated as well and the linear estimator of $S^{n}_{1}$ is given by 
\begin{equation*}
\hat{S}^{n}_{1} = \frac{1}{1+\frac{N_{1}}{\alpha_{1}P}}R^{n}_{11}+\frac{\rho}{1+\frac{\alpha_{1}P+N_{1}}{(1-\alpha_{1})P}}R^{n}_{12}.
\end{equation*}
The distortion resulting in $S_{1}$ is calculated to be
\begin{equation*}
D_{1} = \frac{\sigma^{2}(1-\rho^{2})}{1+\frac{\alpha_{1}P}{N_{1}}}+\frac{\rho^{2}\sigma^{2}}{1+\frac{(1-\alpha_{1})P}{\alpha_{1}P+N_{1}}}.
\end{equation*}
Receiver 2 obtains the estimate of $S^{n}_{2}$ in the same fashion as the independent case by treating the coded portion of the received signal as noise to obtain a distortion
\begin{equation*}
D_{2} = \frac{\sigma^{2}}{1+\frac{(1-\alpha_{1})P}{\alpha_{1}P+N_2}}.
\end{equation*}
Thus 
\begin{align*}
\mathcal{D}_{h} = \Bigg\{(D_{1},D_{2}): D_{1} &\geq \frac{\sigma^{2}(1-\rho^{2})}{1+\frac{\alpha_{1}P}{N_{1}}}+\frac{\rho^{2}\sigma^{2}}{1+\frac{(1-\alpha_{1})P}{\alpha_{1}P+N_{1}}}\\ 
D_{2} &\geq \frac{\sigma^{2}}{1+\frac{(1-\alpha_{1})P}{\alpha_{1}P+N_2}}\Bigg\}.
\end{align*}\\
We now compare the distortion region achieved by the hybrid coding scheme with two possible source channel separation schemes, Scheme A and Scheme B. In Scheme A, $S^{n}_{1}$ and $S^{n}_{2}$ are compressed to obtain messages that can be reliably transmitted over the broadcast channel. The distortion region achieved by this scheme is given by the set
\begin{equation*}
\mathcal{D}_{A} = \Bigg\{(D_{1},D_{2}): D_{1} \geq \frac{\sigma^{2}}{1+\frac{\alpha_{1}P}{N_{1}}}, 
D_{2} \geq \frac{\sigma^{2}}{1+\frac{(1-\alpha_{1})P}{\alpha_{1}P+N_2}}\Bigg\}.
\end{equation*}
The distortion in $S_{1}$ incurred by the hybrid scheme is smaller than the distortion that is achieved by the above source channel separation scheme if
\begin{align*}
&\frac{\sigma^{2}(1-\rho^{2})}{1+\frac{\alpha_{1}P}{N_{1}}}+\frac{\rho^{2}\sigma^{2}}{1+\frac{(1-\alpha_{1})P}{\alpha_{1}P+N_{1}}} < \frac{\sigma^{2}}{1+\frac{\alpha_{1}P}{N_{1}}}\\
\Rightarrow & \frac{P}{N_{1}} < \frac{1-2\alpha_{1}}{\alpha^{2}_{1}}.
\end{align*} 
Thus $\mathcal{D}_{A}\subset \mathcal{D}_{h}$ for $\alpha_{1}<\frac{1}{2}$ and $\frac{P}{N_{1}} < \frac{1-2\alpha_{1}}{\alpha^{2}_{1}}$. \\

In Scheme B, we use the representation $S_{1}^{n} = \rho S_{2}^{n}+V^{n}$, stated earlier in this section. The transmitter compresses $V^{n}$ and $S^{n}_{2}$ to obtain messages that can be reliably communicated to Receivers 1 and 2 respectively. Due to the degraded nature of the broadcast channel, Receiver 1 can mimic Receiver 2 to obtain an estimate of $S^{n}_{2}$. Now Receiver 1, combines the estimates of $S^{n}_{2}$ and $V^{n}$ to construct an estimate of $S^{n}_{1}$. The distortion region thus achieved is given by
\begin{align*}
\mathcal{D}_{B} = \Bigg\{(D_{1},D_{2}): D_{1} &\geq  \frac{\sigma^{2}(1-\rho^{2})}{1+\frac{\alpha_{1}P}{N_{1}}}+\frac{\rho^{2}\sigma^{2}}{1+\frac{(1-\alpha_{1})P}{\alpha_{1}P+N_{2}}}, \\
D_{2} &\geq \frac{\sigma^{2}}{1+\frac{(1-\alpha_{1})P}{\alpha_{1}P+N_2}}\Bigg\}.
\end{align*}
Hence, $\mathcal{D}_{B}\subset \mathcal{D}_{h}$ for $\frac{P}{N}>0$ and $0\leq\alpha_{1}\leq1$. 

\end{proof}
\section{Conclusions}
\label{sec:conclude}  
We present a hybrid coding scheme for source channel communication of correlated Gaussian sources over broadcast channels, that resembles dirty paper coding. We show that the scheme is optimal in terms of achieving the smallest distortion for communicating independent sources. Further, we prove that for a non-trivial set of SNR, the scheme achieves a lower distortion than source channel separation. As a next step, we plan to compare uncoded,  hybrid coded and separately coded transmission schemes to determine regimes where each outperforms the others.

\end{document}